\documentclass[a4paper,runningheads]{llncs}

\usepackage{etoolbox}
\newtoggle{arxiv}
\newtoggle{submission}
\newcommand{\ifarxivelse}[2]{\iftoggle{arxiv}{#1}{\cite[#2]{arxivversion}}}
\toggletrue{arxiv}
\togglefalse{submission}
\iftoggle{submission}{}{
\pagestyle{headings}
\usepackage[paperheight=235mm, paperwidth=155mm,textwidth=12.2cm,textheight=19.3cm,hmarginratio=1:1]{geometry}
}

\usepackage[T1]{fontenc}
\usepackage{amsmath,amssymb}
\usepackage{xspace}
\usepackage{booktabs,multirow}
\usepackage{adjustbox}
\usepackage{cite}
\usepackage{wrapfig}
\usepackage{mathrsfs}
\usepackage{listings}
\usepackage{algorithm, algpseudocode}
\usepackage{graphicx}
\usepackage[%
  bookmarks,
  unicode,
  colorlinks=true,
  allcolors=blue!65!black!90,
  breaklinks=true]{hyperref}
\usepackage{bbding}

\usepackage{enumitem}
\setitemize{noitemsep,topsep=0pt,parsep=0pt,partopsep=0pt,leftmargin=12.0pt}
\setenumerate{noitemsep,topsep=0pt,parsep=0pt,partopsep=0pt}
\setdescription{noitemsep,topsep=2pt,parsep=0pt,partopsep=0pt,leftmargin=12.5pt}

\usepackage{chngcntr}

\makeatletter
\def\orcidID#1{\textsuperscript{\,\smash{\protect\raisebox{-1.25pt}{\href{http://orcid.org/#1}{\protect\includegraphics[scale=.8]{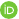}}}}}}
\makeatother

\allowdisplaybreaks %

\usepackage[nameinlink,capitalize,english]{cleveref}
\Crefname{figure}{Fig.}{Figs.}
\crefname{figure}{fig.}{figs.}
\Crefname{tabular}{Tab.}{Tabs.}
\crefname{tabular}{tab.}{tabs.}
\Crefname{section}{Sect.}{Sects.}
\crefname{section}{sect.}{sects.}
\Crefname{appendix}{App.}{Apps.}
\crefname{appendix}{app.}{apps.}
\Crefname{equation}{Eq.}{Eqs.}
\crefname{equation}{eq.}{eqs.}
\creflabelformat{equation}{#2#1#3}
\Crefname{example}{Ex.}{Exs.}
\crefname{example}{ex.}{exs.}

\usepackage{tikz}
\usepackage{tikz-cd}
\tikzset{circled node/.style={circle,draw, inner sep=0, minimum size=1.5em},
nodestyle/.style={circled node,text height=.8em,text depth=.25em}}

\newcommand{\NN}{\ensuremath{\mathbb{N}}\xspace}  %
\newcommand{\RR}{\ensuremath{\mathbb{R}}\xspace}  %

\newcommand{\I}{\ensuremath{\mathscr{I}}\xspace}
\renewcommand{\O}{\ensuremath{\mathscr{O}}\xspace}
\newcommand{\U}{\ensuremath{\mathscr{U}}\xspace}

\newcommand{\1}{\ensuremath{\mathbf{1}}\xspace}
\newcommand{\tool}[1]{\textsc{#1}}
\newcommand{\lang}[1]{\textsc{#1}}

\newcommand{\toolset}{\tool{Modest Toolset}\xspace}

\newcommand{\ie}{i.e.\ }
\newcommand{\etal}{et al.\xspace}

\renewcommand{\iff}{\ensuremath{\Leftrightarrow}\xspace}
\newcommand{\set}[1]{\ensuremath{\{\,#1\,\}}}

\newcommand{\defeq}{\mathrel{\vbox{\offinterlineskip\ialign{\hfil##\hfil\cr{\tiny \rm def}\cr\noalign{\kern0.30ex}$=$\cr}}}}
\renewcommand{\P}{\ensuremath{\mathbf{P}}\xspace}
\newcommand{\M}{\ensuremath{\mathcal{M}}\xspace}

\newcommand{\red}{\ensuremath{\mathit{red}}\xspace}
\newcommand{\abstr}{\ensuremath{\mathit{abstr}}\xspace}
\newcommand{\tail}{\ensuremath{\mathit{tail}}\xspace}
\newcommand{\inp}{\ensuremath{\mathit{inp}}\xspace}
\newcommand{\evt}{\ensuremath{\mathit{evt}}\xspace}
\newcommand{\E}{\mathbb{E}}
\newcommand{\Ut}{\ensuremath{\mathscr{U}\uplus\set{t}}\xspace}

\newcommand{\emphbf}[1]{\textbf{\emph{#1}}}

\makeatletter%
\g@addto@macro\normalsize{%
  \setlength\abovedisplayskip{3pt}%
  \setlength\belowdisplayskip{3pt}%
  \setlength\abovedisplayshortskip{-3pt}%
  \setlength\belowdisplayshortskip{3pt}%
}%
\makeatother

\begin{document}

\title{%
Path Abstraction for Markov Reward Models%
\thanks{
This work was supported
by the EU's Horizon 2020 research and innovation programme under MSCA grant agreement 101008233 (MISSION),
by the Interreg North Sea project STORM\_SAFE,
and
by NWO VIDI grant VI.Vidi.223.110 (TruSTy).
}
}

\author{%
Arnd Hartmanns\orcidID{0000-0003-3268-8674}
\and
Robert Modderman$^{\text{\,(\raisebox{-1.6pt}{\Envelope})}}$\orcidID{0009-0002-9198-3809}
}
\institute{%
University of Twente, Enschede, The Netherlands
$\cdot$ \email{r.modderman@utwente.nl}
}
\authorrunning{A.\ Hartmanns, R.\ Modderman}

\maketitle

\begin{abstract}
Path abstraction originated as a technique for counterexample refinement in probabilistic model checking.
Given a discrete-time Markov chain, it summarises the probabilities passing through a subset of the states onto new transitions of a smaller chain.
In earlier work, we proved its correctness and that it is monotonically absorbing.
In this paper, we extend path abstraction from reachability probabilities on discrete-time Markov chains to expected rewards on Markov reward models.
Working in a novel free monoid view of Markov chains throughout, we prove that path abstraction preserves the Markov reward model structure when abstracting over arbitrary sets of states, and that it remains monotonically absorbing.
Finally, we give a numerical recipe, accompanied by a reference implementation in \mbox{PARI/GP}, that computes path abstraction by solving linear equation systems.
Its correctness rests on the relationship between expected rewards and expected visiting times of transitions.
\end{abstract}

\section{Introduction}

Probabilistic model checking (PMC)~\cite{BAFK18} automatically verifies quantitative properties of systems that exhibit stochastic behaviour, such as randomised algorithms, communication protocols, and fault-tolerant or energy-aware systems.
The most basic formalism that PMC can work on are finite-state discrete-time Markov chains (DTMCs) that step from a state into a discrete probability distribution over successor states.
The most basic PMC query on a DTMC model is for \emph{reachability probabilities}:
the probability of eventually reaching a set of goal states from the initial state.
Annotating the transitions of a DTMC with real-valued rewards yields a \emph{Markov reward model} (MRM).
On MRMs, we can check \emph{expected-reward} properties, for example the expected time to termination, the expected number of messages sent successfully, or the system's expected energy consumption.
Properties can come as queries, asking for the value, or specifications bounding the value by a threshold from above or below.

Counterexamples are a long-standing and important topic in PMC~\cite{ADR08,HKD09}:
when a model violates a specification, a counterexample explains \emph{why}, which is essential for debugging and for building trust in the verdict.
One technique in this area is \emph{path abstraction}, introduced by \'{A}brah\'{a}m \etal~\cite{AJWKB10}, which reroutes the probability mass of all paths through a chosen set of states onto new transitions of a smaller chain.
While created for counterexample refinement~\cite{ADR08}, it also offers an alternative to \emph{state elimination}~\cite{Daw04,HHZ11} for solving a DTMC locally, one set of states at a time.
Whereas state elimination can blow up the number of transitions, path abstraction never increases it---at the cost of the smallest sub-problems remaining linear equation systems, which can however be \emph{much} smaller than for the full DTMC and may therefore be solved quickly, even with exact solvers.

\'{A}brah\'{a}m \etal~\cite{AJWKB10} apply path abstraction recursively over the strongly connected components (SCCs) of the DTMC.
Their paper left some properties that may seem obvious but are crucial for the technique's correctness unproven, in particular that the result of path abstraction is again a DTMC.
In earlier work~\cite{HM25}, we supplied these proofs, and introduced and proved \emph{monotonic absorption}: abstracting over a subset first and then over a superset gives the same result as abstracting over the superset directly.
In the process, we also generalised the technique from abstracting over SCCs to arbitrary sets of non-absorbing states.
A key ingredient of our work was a novel view of DTMCs as structures on the free monoid over their state space, which shortened the definitions and made them more intuitive, and simplified the proofs.
All work on path abstraction, however, was so far confined to unbounded reachability probabilities.

\paragraph{Contributions.}
In this paper, we extend path abstraction from reachability probabilities to \emph{expected rewards} (ERs): concretely, to the expected accumulated reward to reach a goal, one of the standard PMC queries~\cite{HJQW26}.
We prove that path abstraction over arbitrary abstraction sets preserves the MRM structure and satisfies monotonic absorption.
Throughout, we work in the free monoid view, for which we introduce a free monoid definition of MRMs.

Finally, we give a numerical recipe for computing the path abstraction of an MRM over a given abstraction set, with a high-level reference implementation in the computer algebra system \lang{PARI/GP}~\cite{PARI2}.
It translates the abstract formulation into concrete matrix problems---solving linear equation systems, as in PMC for DTMCs.
Closely related to ERs are the \emph{expected visiting times} (EVTs)~\cite{MKQW24} of transitions:
the ER of an MRM equals the sum %
of each transition's EVT multiplied by its reward.
Using this connection, we are able to prove that the recipe's concrete computation matches the abstract formulation of path abstraction.

\section{Background}
\label{sec:Background}

We first introduce free monoids, languages, and basic notions on these that are borrowed from combinatorics on words.\!\footnote{This part of the background is based on~\cite[Sect.\ 2]{HM25},
with small updates to the notation for clarity.
Notably, we use $L_\perp$ to indicate the prefix-minimal part of a language $L$ (instead of $L^\leqslant$), and we use $\red$ to indicate the function that takes a set of letters and a word and outputs a reduced version of its second argument with respect to its first, which was introduced in~\cite{HM25} with the notation of a minus sign.}
Next, we define Markov chains and Markov reward models via the free monoid on the state space.

Throughout, the set of natural numbers \emph{including} $0$ is $\NN=\set{0,1,2,\ldots}$ and \emph{excluding} $0$ is $\NN^+=\set{1,2,\ldots}$.
For a set $X$, we write $2^X$ for its power set.
Given function $f\colon X \to Y$ and $A \subseteq X$, we write $f(A)$ for the image $\set{ f(x) \colon x \in A }$. %

\subsection{Languages and Combinatorics on Words}
We first give a recap on languages and combinatorics on words, with some notions we need for our reasoning with path abstraction for MRMs.
For the definitions, we use terminology from~\cite{Lot02}, a standard work on combinatorics on words and theory of languages.
Throughout, let $\Sigma$ be a non-empty set.
\begin{definition}%
\label{def:free-monoid-languages}
    The set of all finite sequences with symbols in $\Sigma$ is denoted by $\Sigma^*$ and is called the \emphbf{free monoid} \emph{on} $\Sigma$.
    Any subset $L$ of $\Sigma^*$ is called a \emphbf{language} over \emph{alphabet} $\Sigma$.
    Elements of $\Sigma^*$ are called \emph{words} over $\Sigma$.
    The symbols of a word $x\in\Sigma^*$ as elements of $\Sigma$ are called the \emph{letters} of $x$.
    The \emph{empty word} of $\Sigma^*$, \ie the empty sequence, is denoted by $\varepsilon$.
    A word $x\in\Sigma^*$ can be seen as the singleton language $\set{x}$.
    Words of length 1 can and will be identified with letters.
\end{definition}
Given two words $x,y\in\Sigma^*$, we denote the \textbf{concatenation} of $x$ \emph{with} $y$ by $xy$.
Given two languages $K,L\subseteq\Sigma^*$, we can form the concatenation of $K$ and $L$ as $KL\defeq\set{xy\colon x\in K\land y\in L}$.
The \textbf{length} of a word $x\in\Sigma^*$ is denoted by $|x|$.

\begin{definition}
\label{def:path-concatenation}
    For $x\in\Sigma^*$, let $\tail(x)$ denote its \emphbf{tail}, \ie we set $\tail(\varepsilon)\defeq\varepsilon$ and for $a\in\Sigma$ and $x\in\Sigma^*$, we set $\tail(ax)\defeq x$.
    Given $K,L\subseteq\Sigma^*$, we define $K\star L= K(\tail(L))$ and call $K\star L$ the \emphbf{path concatenation} of $K$ with $L$.
\end{definition}
The reason why $\star$ is called \emph{path concatenation} is because of its relation to graph theory: if e.g. in a graph $G=(V,E)$, $\pi v$ and $v\pi'$ are two paths with $\pi,\pi'\in V^*$ and $v\in V$, then we obtain their path concatenation as $\pi v\star v\pi'=\pi v\pi'$.

\begin{definition}
Let $x,x'\in\Sigma^*$. Then, $x'$ is a \emphbf{prefix} of $x$, denoted by $x'\leqslant x$, if there exists $y\in\Sigma^*$ such that $x=x'y$.
If $x'\leqslant x$ but $x'\neq x$, then $x'$ is a \emph{strict prefix} of $x$, which is denoted by $x'<x$.
The order $\leqslant$ on $\Sigma^*$, which is a partial order, is the \emphbf{prefix order} on $\Sigma^*$.
\end{definition}

\begin{definition}
    Given $L\subseteq\Sigma^*$ and $n\in\NN$, we can define its $n$\emphbf{-th power} $L^n$ as the $n$-fold concatenation $LL\cdots L$ of $L$ with itself,
    where $L^0\defeq\set{\varepsilon}$.
    We let $L^*$ denote the \emphbf{Kleene-Star} or \emph{Kleene closure} of $L$, defined by $L^*=\bigcup_{n\in\NN}L^n$; and we define $L^+=\bigcup_{n\in\NN^+}L^n$.
\end{definition}

\begin{definition}
\label{def:prefix-minimal-part}
    Given a language $L\subseteq\Sigma^*$, we can form its \emphbf{prefix-minimal part} $L_\perp\defeq\set{x\in L\colon x'<x\Rightarrow x'\notin L}$.
    We say that a language is a \emphbf{prefix code} if it coincides with its prefix-minimal part; \ie $L$ is a prefix code if $L=L_\perp$.
\end{definition}

\begin{definition}
    Let $x,x'\in\Sigma^*$. Then, $x'$ is a \emphbf{factor} of $x$, denoted by $x'\sqsubseteq x$, if there exist $y,z\in\Sigma^*$ such that $x=yx'z$.
\end{definition}

\begin{definition}
\label{def:factor-occurrences}
    Let $x,x'\in\Sigma^*$.
    Then, we let $|x|_{x'}$ denote the number of \emphbf{factor occurrences} in $x$:
    $|x|_{x'}\defeq|\set{(y,z)\in\Sigma^*\times\Sigma^*\colon x=yx'z}|$.\!\footnote{Note that $|x|_{x'}\in\NN$ as $x$ is finite, $x'\sqsubseteq x\Leftrightarrow |x|_{x'}>0$,
and $|x|=\sum_{a\in\Sigma}|x|_a$.}
\end{definition}

\begin{example}
    We have $|aabab|_{ab}=2$, $|aaabaababb|_{aa}=3$, $|x|_\varepsilon=|x|+1$, $x'\leqslant x\Rightarrow |x|_{x'}>0$, and $|a^n|_{a^m}=\max(0,n-m+1)$.
\end{example}
Now, we define a function $\red$, which we call the \textbf{reduction operator}, that takes a word $x$ and a set $A\subseteq\Sigma$ of letters and outputs a ``reduced'' version $\red(A,x)$ of $x$ relative to $A$.
In~\cite{HM25}, $\red(A,x)$ was written as $x-A$.

\begin{definition}
    We define $\red\colon2^\Sigma\times\Sigma^*\to\Sigma^*$ inductively with respect to its second argument with two base cases. Given $A\subseteq\Sigma$, we set
    $\red(A,\varepsilon)\defeq\varepsilon$, $\red(A,a)\defeq a$ for all $a\in\Sigma$,
    and for all $a,b\in\Sigma$ and $x\in\Sigma^*$ we set
    \begin{equation}
        \red(A,abx)=\begin{cases}
            \red(A,ax),&\text{if }a\in A\land b\in A,\\
            a(\red(A,bx)),&\text{otherwise}.
        \end{cases}
    \end{equation}
    We write $\red(A)$ for the function $x\mapsto\red(A,x)$ in $\Sigma^*\to\Sigma^*$.
    For $L\subseteq\Sigma^*$ we write $\red(A)^{-1}(L)$ for the inverse image $\set{x\in\Sigma^*\colon\red(A,x)\in L}$ of $L$.
\end{definition}
Given $A$, $\red(A, x)$ takes all factors of $x$ in $A^*$ that are maximal in length and replaces these with their respective first letters. In particular, $\red(\Sigma)$ reduces any non-empty $x\in\Sigma^*$ to its first letter and leaves $\varepsilon$ unchanged.
Note furthermore that $\red(A,x)$ does not contain two consecutive symbols of $A$.

\begin{example}
    Over $\Sigma=\set{a,b,c,\ldots,z}$,
    we have $\red(\set{a,b})(a^+bc)=ac$,\\
    $\red(\set{p,i,n,g},mapping)=map$,
    and $\red(\varnothing,x)=x$ for all $x\in\Sigma^*$.
\end{example}

\begin{example}
\label{ex:red-2state-char}
We can characterise $L'=\red(A)^{-1}(L)$ for the case $L$ is a word $ab$ of length 2.
Let $A\subseteq\Sigma$. Then, we have $L'=\varnothing$ if $a,b\in A$; $L'=aA^*b$ if $a\in A\not\ni b$; $L'=abA^*$ if $a\notin A\ni b$; and $L'=\set{ab}$ if $a,b\notin A$.
\end{example}
We restate two key results on the function $\red$ from~\cite{HM25}, namely (a slight extension of)~\cite[Thm. 1.2]{HM25} and~\cite[Thm. 1.4]{HM25}.\!\footnote{The proofs also hold for the case $\Sigma$ is infinite: the assumption in~\cite[Thm 1.2~\&~1.4]{HM25} that $\Sigma$ should be finite was never used and thus is superfluous.}

\begin{theorem}
\label{thm:inverse-images-red-perp-star}
    Let $A\subseteq\Sigma$, $x,y\in\Sigma^*$, and $a\in\Sigma$. Then, we have
    \begin{equation}
    \label{eq:inverse-images-red-perp-star-id}
        (\red(A)^{-1}(xa))_\perp\star(\red(A)^{-1}(ay))_\perp=(\red(A)^{-1}(xay))_\perp,
    \end{equation}
    \begin{equation}
    \label{eq:inverse-images-red-perp-star-link}
        (\red(A)^{-1}(xa))_\perp\subseteq \Sigma^*a\text{, and }
        (\red(A)^{-1}(ay))_\perp\subseteq a\Sigma^*.
    \end{equation}
    \begin{proof}
        We note that~\eqref{eq:inverse-images-red-perp-star-id} is immediate from~\cite[Thm. 1.2]{HM25}, and~\eqref{eq:inverse-images-red-perp-star-link} is established straightforwardly using the definitions of $\red$ and $\perp$.
    \end{proof}
\end{theorem}

\begin{theorem}
\label{thm:inverse-images-red-union}
    Let $A\subseteq B\subseteq\Sigma$. Then, for all $x\in\Sigma^*$ we have
    \begin{equation}
        (\red(B)^{-1}(x))_\perp=\biguplus_{y\in(\red(B)^{-1}(x))_\perp\land\red(A,y)=y}(\red(A)^{-1}(y))_\perp.
    \end{equation}
\end{theorem}

\subsection{Discrete-Time Markov Chains and Markov Reward Models}

We start with a standard automata-based definition of discrete-time Markov chains.
Throughout, we let $S$ be a finite, non-empty set.

\begin{definition}
\label{def:DTMC}
    A \emphbf{discrete-time Markov chain} (DTMC) with \emph{state space} $S$ is a triple $M=(S,s_0,\P)$ where $s_0\in S$ is called the \emph{initial state} of $M$, and $\P\colon S\times S\to[0,1]$ is a function satisfying $\sum_{t\in S}\P(s,t)=1$ for all $s\in S$, called the \emph{transition probability function} of $M$.
\end{definition}
Within PMC, the requirement that $\sum_{t\in S}\P(s,t)=1$ holds for all $s\in S$ is often too strict, and is in many cases relaxed to $\sum_{t\in S} \P(s,t)\leqslant1$, which yields the definition of \emph{substochastic} DTMC.
In this paper, we also do this and henceforth use the term \textbf{Markov chain} to refer to substochastic DTMCs.

\begin{definition}
\label{def:Markov-chain}
    For a Markov chain $M=(S,s_0,\P)$, a pair $(s,t)\in S \times S$ is a \emphbf{transition} in $M$ if $\P(s,t)>0$.
    A finite sequence $\pi=(s_1,s_2)(s_2,s_3)\ldots(s_l,s_{l+1})$ of $l\geqslant1$ transitions is a \emphbf{path} in $M$.
    The \emph{probability} of $\pi$ in $M$ is defined as the product of the probabilities of its transitions, namely $\prod_{i=1}^l\P(s_i,s_{i+1})$.
\end{definition}
It would be useful if we could denote the probability of path $\pi$ simply by $\P(\pi)$.
By identifying a path in $(S\times S)^*$ of the form $(s_1,s_2)(s_2,s_3)\cdots(s_l,s_{l+1})$
with the sequence in $S^*$ of the form $s_1\cdots s_{l+1}$, we can assign the probability of the path to the corresponding sequence (and probability 0 to sequences that are not paths) and so extend $\P$ to $\set{x\in S^*\colon|x|\geqslant2}\to[0,1]$.
To extend $\P$ to $S^+\to[0,1]$, we also set $\P(s)\defeq1$ for all $s\in S$.
For all $x,y\in S^*$ and $s\in S$ we then have $\P(xsy)=\P(xs\star sy)=\P(xs)\P(sy)$,
but also that with this rule one can show by induction on $l$ that $\P(s_1s_2\cdots s_{l+1})=\prod_{i=1}^l\P(s_is_{i+1})$ for all $l\in\NN^+$.

With this alternative notation for transition probability functions, i.e. those that map sequences contained in the \emph{free monoid} on the state space $S$ to $[0,1]$ directly, we can define Markov chains in the free monoid setting.

\begin{definition}
\label{def:Markov-chain-free-monoid}
    A \emphbf{Markov chain} on \emph{state space} $S$ is a triple $M=(S,s_0,\P)$ where $s_0\in S$ is called the \emph{initial state} of $M$, and $\P\colon S^+\to[0,1]$ is a function satisfying $\sum_{t\in S}\P(st)\leqslant1$ and $\P(s)=1$ for all $s\in S$, and $\P(xsy)=\P(xs)\P(sy)$ for all $x,y\in S^*$ and $s\in S$, called the \emph{transition probability function} of $M$.
\end{definition}
A function $\P\colon S^+\to[0,1]$ qualifying as a transition probability function in~\Cref{def:Markov-chain-free-monoid} can be extended to $2^{S^+}\to[0,\infty)$ by setting, for all $L\subseteq S^+$,
\begin{equation}
\label{eq:prob-mass}
    \P(L)\defeq\sum_{x\in L_\perp}\P(x),
\end{equation} which we call the \textbf{probability mass} of $L$ \emph{with respect to} $\P$.\!\footnote{The reason why we do not distinguish between languages with coinciding prefix-minimal parts is that the \emph{event} $x'\in S^+$ is implied by the event $x$ whenever $x'\leqslant x$.}

\begin{definition}
\label{def:MRM}
    A \emphbf{Markov reward model} (MRM) is a tuple $(S,s_0,\P,R)$ where $(S,s_0,\P)$ is a Markov chain as in~\Cref{def:Markov-chain-free-monoid} and $R$ is a function $S^+\to\RR$ satisfying $R(x)=0$ for all $x\in S^+$ for which $|x|=1$ or $\P(x)=0$,
    and $R(xsy)=R(xs)+R(sy)$ for all $x,y\in S^*$ and $s\in S$ for which $\P(xsy)>0$.
    The function $R$ is called a \emphbf{reward function} for the Markov chain $(S,s_0,\P)$.
\end{definition}
This free-monoid definition of MRMs is new.
As path abstraction works on transitions, we w.l.o.g.\ associate rewards to transitions, not states, here.

\begin{remark}
\label{rem:correspondence-free-monoid-classical}
    A Markov chain or MRM in the free monoid setting corresponds to exactly one in the classical setting---\ie where $\P$ has signature $S\times S\to[0,1]$ and $R\colon S \times S \to \RR$.
    This is due to the fact that in the free monoid setting, we let $\P$ and $R$ carry the inherent structure that its values on $xsy$ (where $x,y\in S^*$ and $s\in S$) can be expressed in terms of those on $xs$ and $sy$---multiplicatively and additively, respectively.
\end{remark}

\subsection{Probabilistic Model Checking}
\label{sec:PMC}

Given a set of goal states $G \subseteq S \setminus\set{s_0}$, we are interested in the probability of reaching $G$ from $s_0$, and in the expected reward accumulated before reaching~$G$~\cite[Sect.\ 10]{BK08}.
We w.l.o.g.\ assume that $s \in G \iff \P(ss)=1$, \ie the goal states are exactly the \emph{absorbing} states.
Let $A = \set{s\in S\colon\P(ss)<1} = S \setminus G$.
Then
$$\P^{s_0}(\diamondsuit\, t) \defeq \P(s_0A^*t)$$
is the \textbf{reachability probability} of $t$, and
$$\E^{s_0}(\diamondsuit\, t)\defeq\textstyle\sum_{x\in s_0A^*t}\P(x)R(x)$$
is the \textbf{expected reward} (ER) to reach $t$.
Note that paths not reaching $t$ are not considered in the sum, so they effectively contribute reward $0$ to the ER.\footnote{There are three common ways to handle paths that do not reach the goal in PMC: assign reward $0$, reward $\infty$, or the \emph{infinite} path's cumulative reward~\cite[Sect.\ 4.1]{Sim14}.}
\textbf{Solving} or (fully) \textbf{model checking} an MRM %
means calculating the reachability probability and ER for all paths from $s_0$ to \emph{any} absorbing~state.

\section{Path Abstraction}

We recall path abstraction as defined in~\cite[Def. 6]{HM25}, then we extend it to MRMs.

\subsection{Path Abstraction for Markov chains and MRMs}

Let us first reformulate~\cite[Def.\ 6]{HM25} using our slightly updated notation.

\begin{definition}
\label{def:Markov-chain-path-abstr}
    Given $A\subseteq S$ and a Markov chain $M=(S,s_0,\P)$, we define the \emphbf{path abstraction} $\abstr(A,M)$ of $M$ \emph{with respect to} $A$ as follows.
    We set $\inp(M,A)\defeq\set{s\in A\colon s=s_0\lor\P((S\setminus A)s)>0}$ which we call the \emph{input states} of $A$ with respect to $M$---\ie those states of $A$ that are reachable from outside in one step, including $s_0$ if $s_0\in A$.
    Then, we set $\abstr(A,M)\defeq(S,s_0,\P')$ where
    we define $\P'\colon S^+\to[0,1]$ by
    \begin{equation}
    \label{eq:Markov-chain-path-abstr}
        \P'(x)\defeq\begin{cases}
            1,&\text{if }|x|=1,\\
            \P(\red(A)^{-1}(x)),&\text{if }|x|\geqslant2\land x\in(S\setminus(A\setminus\inp(M,A)))^*,\\
            0,&\text{otherwise}.
        \end{cases}
    \end{equation}
\end{definition}
Intuitively, path abstraction removes all transitions that come from or lead to a non-input state in $A$, and instead adds new transitions that summarise the behaviour inside $A$ directly connecting the input states to the output states of $A$ (\ie those outside $A$ that had an incoming transition from $A$).
Now for MRMs:

\begin{definition}
\label{def:MRM-path-abstr}
    Given $A\subseteq S$ and an MRM $\M=(S,s_0,\P,R)$, we define the \emphbf{path abstraction} $\abstr(A,\M)$ of $\M$ \emph{with respect to} $A$ as follows.
    Writing $M=(S,s_0,\P)$ for the underlying Markov chain of $\M$ and $(S,s_0,\P')$ for the path abstraction $\abstr(A,M)$ of $M$ with respect to $A$ as defined in~\Cref{def:Markov-chain-path-abstr}, we set $\abstr(A,\M)\defeq(S,s_0,\P',R')$ where we define $R'\colon S^+\to\RR$ through
    \begin{equation}
    \label{eq:reward-function}
        R'(x)\defeq\begin{cases}
            \frac{1}{\P'(x)}\sum_{y\in(\red(A)^{-1}(x))_\perp}\P(y)R(y)
            &\text{if }|x|>1\land\P'(x)>0\\
            0&\text{otherwise}.
        \end{cases}
    \end{equation}
\end{definition}
On an MRM, path abstraction takes the \emph{expected} reward (\ie rewards multiplied by probabilities, summed) corresponding to passing through the abstraction set from an input to an output state, \emph{normalises} it by dividing by the (internal) probability of this set of paths, and uses the result as the reward for the new transition.
Moving these \emph{normalised} ERs, instead of the full ERs, onto the new transitions is the key ingredient as to why our new path abstraction works for ERs and satisfies the desired monotonic absorption property. %
\begin{example}
    Consider MRM $\M$ in~\Cref{subfig:M} on state space $S=\set{s_1,\ldots,s_8}$. %
    A transition's probability $p$ and reward $r$ are written as $p;r$.
    Let us abstract $\M$ over $A = \set{s_2,s_5,s_6}$.
    The single input state is $s_2$, and the output states are $s_3$ and $s_8$.
    Effectively, to perform the path abstraction is then to locally model-check the sub-chain with state space $A \cup \set{s_3, s_8}$ indicated with dashed lines in~\Cref{subfig:M}.
    The result is the (smaller) chain with state space $\set{s_2,s_3,s_8}$ embedded in $\M_1$ in~\Cref{subfig:M-1}.
    We further abstract over $\set{s_3, s_4}$ to obtain $\M_{1,2}$ of~\Cref{subfig:M-12}, and then over $\set{s_1, s_2, s_3}$ to end up with $\M_{1,2,3}$ of~\Cref{subfig:M-123}.
    Here, we can directly read off the probability $\frac{5}{9}$ of reaching $s_7$, or the ER $\frac{5}{9} \cdot \frac{142}{5} + \frac{4}{9} \cdot 25 = \frac{242}{9}$ to reach $\set{s_7, s_8}$.
\end{example}

\begin{figure}[t]
\centering
\vspace{-1em}
\begin{minipage}[b]{0.5\textwidth}
\centering
\begin{tikzpicture}
\draw
(0.3, 1) node(ghost1){}
(1, 1) node[circled node](s1){$s_1$}
(2.4, 1) node[circled node, fill=black!15](s2){$s_2$}
(1, 2.4) node[circled node](s3){$s_3$}
(2.4, 2.4) node[circled node](s4){$s_4$}
(2.4, -0.4) node[circled node,  fill=black!15](s5){$s_5$}
(3.8, -0.4) node[circled node,  fill=black!15](s6){$s_6$}
(3.8, 2.4) node[circled node](s7){$s_7$}
(3.8, 1) node[circled node](s8){$s_8$};
\draw
(4.8, 2.4) node{$1;2$}
(4.8, 1.0) node{$1;4$}
(4.1, 0.3) node{$\frac{1}{4};8$}
(3.2, -0.9) node{$\frac{1}{2};2$}
    (2.8, 0.0) node{$1;2$}
(3.15, 0.7) node{$\frac{1}{4};1$}
(3.5, 1.8) node{$\frac{1}{12};3$}
(3.1, 2.7) node{$\frac{1}{6};4$}
(1.4, 2.85) node{$\frac{3}{4};5$}
(1.8, 1.95) node{$1;1$}
(2.1, 0.3) node{$\frac{1}{3};3$}
(1.4, 1.6) node{$\frac{2}{3};2$}
(0.7, 1.7) node{$\frac{1}{6};3$}
(1.6, 0.7) node{$\frac{5}{6};2$};
\draw[-stealth] (ghost1) -- (s1);
\draw[-stealth] (s1) to[out = 0, in = 180] (s2);
\draw[-stealth] (s1) -- (s3);
\draw[-stealth] (s2) -- (s3);
\draw[-stealth] (s2) -- (s5);
\draw[-stealth] (s4) -- (s7);
\draw[-stealth] (s4) -- (s8);
\draw[-stealth] (s6) -- (s2);
\draw[-stealth] (s3) to[out = -30, in = -150] (s4);
\draw[-stealth] (s4) to[out = 150, in = 30] (s3);
\draw[-stealth] (s5) to[out = 30, in = 150] (s6);
\draw[-stealth] (s6) to[out = -150, in = -30] (s5);
\draw[-stealth] (s6) -- (s8);
\draw[-stealth] (s7) to[out = 30, in = -30, looseness = 8] (s7);
\draw[-stealth] (s8) to[out = 30, in = -30, looseness = 8] (s8);
\draw[densely dotted, line width=0.6pt]
(1.325, -1.2)--(5.2, -1.2);
\draw[densely dotted, line width=0.6pt]
(5.2,-1.2)--(5.2,1.4);
\draw[densely dotted, line width=0.6pt]
(5.2,1.4)--(2.8,1.4);
\draw[densely dotted, line width=0.6pt]
(2.8,1.4)--(1.8,2.4);
\draw[densely dotted, line width=0.6pt]
(1.8,2.4)--(1.8,3.2);
\draw[densely dotted, line width=0.6pt]
(1.8,3.2)--(0.2,3.2);
\draw[densely dotted, line width=0.6pt]
(0.2,3.2)--(0.2,1.325);
\draw[densely dotted, line width=0.6pt]
(0.2,1.325)--(1.325,1.325);
\draw[densely dotted, line width=0.6pt]
(1.325,1.325)--(1.325,-1.2);
\end{tikzpicture}
\vspace{-6pt}
\caption{MRM $\M$}
\label{subfig:M}
\end{minipage}\hfill
\begin{minipage}[b]{0.49\textwidth}
\centering
\begin{tikzpicture}
\draw
(0.3, 1) node(ghost1){}
(1, 1) node[circled node](s1){$s_1$}
(2.4, 1) node[circled node](s2){$s_2$}
(1, 2.4) node[circled node, fill=black!15](s3){$s_3$}
(2.4, 2.4) node[circled node, fill=black!15](s4){$s_4$}
(3.8, 2.4) node[circled node](s7){$s_7$}
(3.8, 1) node[circled node](s8){$s_8$};
\draw
(2.95, 0.7) node{$\frac{1}{5};19$}
(4.8, 2.4) node{$1;2$}
(4.8, 1.0) node{$1;4$}
(3.5, 1.8) node{$\frac{1}{12};3$}
(3.1, 2.7) node{$\frac{1}{6};4$}
(1.4, 2.85) node{$\frac{3}{4};5$}
(1.8, 1.95) node{$1;1$}
(1.5, 1.5) node{$\frac{4}{5};4$}
(0.7, 1.5) node{$\frac{1}{6};3$}
(1.6, 0.7) node{$\frac{5}{6};2$};
\draw[-stealth] (ghost1) -- (s1);
\draw[-stealth] (s1) to[out = 0, in = 180] (s2);
\draw[-stealth] (s1) -- (s3);
\draw[-stealth] (s2) -- (s3);
\draw[-stealth] (s4) -- (s7);
\draw[-stealth] (s4) -- (s8);
\draw[-stealth] (s3) to[out = -30, in = -150] (s4);
\draw[-stealth] (s4) to[out = 150, in = 30] (s3);
\draw[-stealth] (s7) to[out = 30, in = -30, looseness = 8] (s7);
\draw[-stealth] (s8) to[out = 30, in = -30, looseness = 8] (s8);
\draw[-stealth] (s2)--(s8);
\draw (0.3,-1.1) node(ghost2){};
\draw[densely dotted, line width=0.6pt]
(5.2,3.2)--(0.4,3.2);
\draw[densely dotted, line width=0.6pt]
(0.4,3.2)--(0.4,1.75);
\draw[densely dotted, line width=0.6pt]
(0.4,1.75)--(2.5,1.75);
\draw[densely dotted, line width=0.6pt]
(2.5,1.75)--(3.75,0.5);
\draw[densely dotted, line width=0.6pt]
(3.75,0.5)--(5.2,0.5);
\draw[densely dotted, line width=0.6pt]
(5.2,0.5)--(5.2,3.2);
\end{tikzpicture}
\vspace{-6pt}
\caption{$\M_1=\abstr(\set{s_2,s_5,s_6},\M)$}
\label{subfig:M-1}
\end{minipage}\\[8pt]
\begin{minipage}[b]{0.49\textwidth}
    \centering
    \begin{tikzpicture}
    \draw
    (0.3, 1) node(ghost1){}
    (1, 1) node[circled node, fill=black!15](s1){$s_1$}
    (2.4, 1) node[circled node, fill=black!15](s2){$s_2$}
    (2.4, 2.4) node[circled node, fill=black!15](s3){$s_3$}
    (3.8, 2.4) node[circled node](s7){$s_7$}
    (3.8, 1) node[circled node](s8){$s_8$};
    \draw
    (3.1, 0.7) node{$\frac{1}{5};19$}
    (4.8, 2.4) node{$1;2$}
    (4.8, 1.0) node{$1;4$}
    (3.1, 2.7) node{$\frac{2}{3};23$}
    (3.4, 1.9) node{$\frac{1}{3};22$}
    (1.45, 1.9) node{$\frac{1}{6};3$}
    (2.1, 1.6) node{$\frac{4}{5};4$}
    (1.6, 0.7) node{$\frac{5}{6};2$};
    \draw[-stealth] (ghost1) -- (s1);
    \draw[-stealth] (s1) to[out = 0, in = 180] (s2);
    \draw[-stealth] (s1) -- (s3);
    \draw[-stealth] (s2) -- (s3);
    \draw[-stealth] (s7) to[out = 30, in = -30, looseness = 8] (s7);
    \draw[-stealth] (s8) to[out = 30, in = -30, looseness = 8] (s8);
    \draw[-stealth] (s2)--(s8);
    \draw[-stealth] (s3)--(s8);
    \draw[-stealth] (s3)--(s7);
    \draw[densely dotted, line width=0.6pt]
    (5.2,3)--(0.2,3);
    \draw[densely dotted, line width=0.6pt]
    (0.2,3)--(0.2,0.4);
    \draw[densely dotted, line width=0.6pt]
    (0.2,0.4)--(5.2,0.4);
    \draw[densely dotted, line width=0.6pt]
    (5.2,0.4)--(5.2,3);
\end{tikzpicture}
\vspace{-6pt}
\caption{$\M_{1,2}=\abstr(\set{s_3,s_4},\M_1)$}
\label{subfig:M-12}
\end{minipage}\hfill
\begin{minipage}[b]{0.49\textwidth}
    \centering
    \begin{tikzpicture}
    \draw
    (0.3, 1) node(ghost1){}
    (1, 1) node[circled node](s1){$s_1$}
    (2.4, 2.4) node[circled node](s7){$s_7$}
    (2.4, 1) node[circled node](s8){$s_8$};
    \draw
    (3.4, 2.4) node{$1;2$}
    (3.4, 1.0) node{$1;4$}
    (1.3, 1.9) node{$\frac{5}{9};\frac{142}{5}$}
    (1.6, 0.7) node{$\frac{4}{9};25$};
    \draw[-stealth] (ghost1) -- (s1);
    \draw[-stealth] (s1) to[out = 0, in = 180] (s2);
    \draw[-stealth] (s1) -- (s3);
    \draw[-stealth] (s7) to[out = 30, in = -30, looseness = 8] (s7);
    \draw[-stealth] (s8) to[out = 30, in = -30, looseness = 8] (s8);
\end{tikzpicture}
\vspace{-6pt}
\caption{$\M_{1,2,3}=\abstr(\set{s_1,s_2,s_3},\M_{1,2})$}
\label{subfig:M-123}
\end{minipage}
\end{figure}

\subsection{Path Abstraction as Probabilistic Model Checking}

Recall the PMC problem of \Cref{sec:PMC}:
Given the absorbing states as goal states, compute $\P^{s_0}(\diamondsuit\, t)$ or $\E^{s_0}(\diamondsuit\, t)$ for all such states~$t$.
Thus, to perform PMC via path abstraction, we only take abstraction sets not containing absorbing states.
Path abstraction for MRMs then stores exactly these quantities in the (fully) path-abstracted MRM $\M'=\abstr(\set{s\in S\colon\P(ss)<1},\M)$.
Indeed: by~\Cref{def:Markov-chain-path-abstr,def:MRM-path-abstr}, for every absorbing $t$ with $\P^{s_0}(\diamondsuit\,t)>0$ we store $\P^{s_0}(\diamondsuit\,t)$ as a probability and $\E^{s_0}(\diamondsuit\, t)/\P^{s_0}(\diamondsuit\,t)$ as a reward onto the transition $s_0\to t$ in $\M'$.

To prove correctness of path abstraction as a means to fully model check an MRM with intermediate MRM, we thus wish to prove that path abstracting (1)~an MRM outputs another MRM, and (2)~preserves reachability probabilities and ERs.
In \Cref{sec:path-abstr} below, we establish (2) by proving in~\Cref{thm:MRM-monab} that path abstraction for MRMs is monotonically absorbing, after establishing (1), the analogue of which for Markov chains was done in~\cite{HM25}.
The fact that we can obtain intermediate MRMs then enables us to model check an MRM with \emph{counterexample refinement}.

\begin{example}
A \emph{counterexample} to the formula $\E^{s_0}(\diamondsuit s_7)\leqslant10$ to reach $s_7$ in $\M$ in~\Cref{subfig:M} is given by the path $s_1s_2s_3s_7$ in the intermediate $\M_{1,2}$, where $\M_{1,2,3}$ is the resulting MRM of fully model checking $\M$.
The probability of $s_1s_2s_3s_7$ in $\M_{1,2}$ is $\frac{5}{6}\cdot\frac{4}{5}\cdot\frac{2}{3}=\frac{4}{9}$ and has reward $2+4+23=29$, and thus the ER to reach $s_7$ in $\M_{1,2}$ is $\geqslant\frac{4}{9}\cdot 29>10$ and hence the ER to reach $s_7$ in the original MRM $\M$ is $>10$ as well, as path abstraction preserves total expected accumulated rewards.
\end{example}

\section{Correctness Properties}
\label{sec:path-abstr}

We prove that path abstraction is well-defined and monotonically absorbing.

\subsection{Well-definedness}

We first prove that path-abstracting an MRM over any subset of its state space yields another MRM.
The idea of the proof is that we can break up the underlying language of the path according to its two fragments using~\Cref{thm:inverse-images-red-perp-star}, and make use of the fact that in~\Cref{def:MRM-path-abstr} we store \emph{normalised} local ERs. %

\begin{theorem}
    \label{thm:well-defined}
    Let $\M=(S,s_0,\P,R)$ be an MRM, and $A\subseteq S$.
    Then, $\abstr(A,\M)$ is also an MRM.
\end{theorem}

\begin{proof}
Write $(S,s_0,\P',R')=\abstr(A,\M)$.
Then, $(S,s_0,\P')$ is a Markov chain by~\cite[Lemma 2]{HM25}:
We are thus left with establishing that $R'$ satisfies the defining properties in~\Cref{def:MRM} of a reward function for the Markov chain $(S,s_0,\P')$.

\emph{Proof that $R'(x)=0$ for all $x\in S^+$ with $|x|=1$ or $\P'(x)=0$: }for all such $x$, we are in the second case of~\eqref{eq:reward-function} which automatically gives $R'(x)=0$.

\emph{Proof that $R'(xsy)=R'(xs)+R'(sy)$ for all $x,y\in S^*$ and $s\in S$ for which $\P'(xsy)>0$: }let $x,y\in S^*$ and $s\in S$ be such that $\P'(xsy)>0$.
If $x=\varepsilon$ then $xsy=sy$, and, by~\eqref{eq:reward-function}, $R'(xs)=0$, so $R'(xsy)=R'(xs)+R'(sy)$. Similarly, one proves $R'(xsy)=R'(xs)+R'(sy)$ if $y=\varepsilon$.
Thus, assume that $x,y\in S^+$.
Then, $xsy$, $xs$, and $sy$ have lengths greater than $1$, so for each of these three sequences we can use the first case of~\eqref{eq:reward-function} to write their images under $R'$.
Thus,
\begin{subequations}\label{eq:well-defined}
\begin{align}
    R'(xsy)&=\frac{1}{\P'(xsy)}\sum_{z\in(\red(A)^{-1}(xsy))_\perp}\hspace{-3ex}\P(z)R(z)\label{eq:well-defined-1}
    \\&=\frac{1}{\P'(xsy)}\sum_{u\in(\red(A)^{-1}(xs))_\perp}\sum_{v\in(\red(A)^{-1}(sy))_\perp}\hspace{-3ex}\P(u\star v)R(u\star v)\label{eq:well-defined-2}
    \\&=\frac{1}{\P'(xsy)}\sum_{u\in(\red(A)^{-1}(xs))_\perp}\sum_{v\in(\red(A)^{-1}(sy))_\perp}\hspace{-3ex}\P(u)\P(v)(R(u)+R(v))\label{eq:well-defined-3.1}
    \\&=\frac{1}{\P'(xs)\P'(sy)}\sum_{u\in(\red(A)^{-1}(xs))_\perp}\hspace{-3ex}\P(u)R(u)\sum_{v\in(\red(A)^{-1}(sy))_\perp}\hspace{-3ex}\P(v)\label{eq:well-defined-4.1}
    \\&+\frac{1}{\P'(xs)\P'(sy)}\sum_{v\in(\red(A)^{-1}(sy))_\perp}\hspace{-3ex}\P(v)R(v)\sum_{u\in(\red(A)^{-1}(xs))_\perp}\hspace{-3ex}\P(u)\label{eq:well-defined-4.2}
    \\&=\frac{1}{\P'(sy)}R'(xs)\P'(sy)+\frac{1}{\P'(xs)}R'(sy)\P'(xs)\label{eq:well-defined-5}
    \\&=R'(xs)+R'(sy),
\end{align}
\end{subequations}
where~\eqref{eq:well-defined-1} follows from~\eqref{eq:reward-function}, ~\eqref{eq:well-defined-2} follows from~\Cref{thm:inverse-images-red-perp-star},~\eqref{eq:well-defined-3.1} follows from the fact that for all such $u$ and $v$ we have $\P(u\star v)=\P(u)\P(v)$ and $R(u\star v)=R(u)+R(v)$ because $(S,s_0,\P,R)$ is an MRM and $u\in S^*s$ and $v\in s S^*$ by~\eqref{eq:inverse-images-red-perp-star-link} (where $\Sigma=S$ and $a=s$),~\eqref{eq:well-defined-4.1} and~\eqref{eq:well-defined-4.2} follow from the fact that $\P'(xsy)=\P'(xs)\P'(sy)$ because $(S,s_0,\P')$ is a Markov chain, and~\eqref{eq:well-defined-5} follows from~\eqref{eq:Markov-chain-path-abstr}---as $|xs|,|sy|>1$ and $\P'(xs),\P'(sy)>0$---and~\eqref{eq:reward-function}.
\qed
\end{proof}

\subsection{Monotonic Absorption}

Now, we prove that the monotonic absorption property of path abstraction, which was established to hold for Markov chains as proven in~\cite[Thm.\ 3]{HM25}, also holds for MRMs.
We first repeat a smaller, technical lemma.

\begin{lemma}
\label{lem:spades}
    Let $M=(S,s_0,\P)$ be a Markov chain, and let $A\subseteq B\subseteq S$.
    Let $x\in(S\setminus(B\setminus\inp(M,B)))^*$.
    Let $y\in(\red(B)^{-1}(x))_\perp$ and assume that $\red(A,y)=y$ and $y\notin(S\setminus(A\setminus\inp(M,A)))^*$.
    Then we have $\P(\red(A)^{-1}(y))=0$.
    \begin{proof}
        This is proven in~\cite[App. B]{HM25arxivversion} in the proof of~\cite[Thm. 3]{HM25arxivversion}.
            \end{proof}
\end{lemma}

\begin{theorem}
\label{thm:MRM-monab}
    Let $\M=(S,s_0,\P,R)$ be an MRM, and let $A\subseteq B\subseteq S$.
    Then, we have the \emph{monotonic absorption} property
    \begin{equation}
        \abstr(B,\abstr(A,\M))=\abstr(B,\M).
    \end{equation}
\end{theorem}
The idea of the proof is to exploit~\Cref{thm:inverse-images-red-union} just like was done for monotonic absorption for Markov chains as done in~\cite[Thm. 3]{HM25}, together with again the fact that in~\Cref{def:MRM-path-abstr} we store \emph{normalised} local ERs. %
\begin{proof}
    Write $(S,s_0,\P_A,R_A)=\abstr(A,\M)$, $(S,s_0,\P_B,R_B)=\abstr(B,\M)$,
    and write $(S,s_0,\P_{A,B},R_{A,B})=\abstr(B,\abstr(A,\M))$.
    We have to show that $\P_B=\P_{A,B}$ and $R_B=R_{A,B}$.
    However, $\P_B=\P_{A,B}$ follows directly from the monotonic absorption property of Markov chains as established in~\cite[Thm.\ 3]{HM25}.
    What remains to show is that $R_B=R_{A,B}$.
    To this end, let $x\in S^+$.
    If $|x|=1$ or $\P_{A,B}(x)=\P_B(x)=0$, then by~\Cref{def:MRM-path-abstr} we obtain $R_B(x)=0$ and $R_{A,B}(x)=0$ and thus $R_B(x)=R_{A,B}(x)$.
    We therefore assume that $x\in S^+$ satisfies $|x|>1$ and $\P_{A,B}(x)=\P_B(x)>0$.
    We prove that $\P_{A,B}(x)R_{A,B}(x)=\P_B(x)R_B(x)$ to show that $R_{A,B}(x)=R_B(x)$.
    We have
    \begin{subequations}\label{eq:monab}
    \begin{align}
        \P_{A,B}(x)R_{A,B}(x)
        &=\sum_{y\in(\red(B)^{-1}(x))_\perp}\P_A(y)R_A(y)\label{eq:monab-1}
        \\&=\sum_{\substack{y\in(\red(B)^{-1}(x))_\perp\colon\\\red(A,y)=y\land y\in(S\setminus(A\setminus\inp(M,A)))^*}}\P_A(y)R_A(y)\label{eq:monab-2}
        \\&=\sum_{\substack{y\in(\red(B)^{-1}(x))_\perp\colon\\\red(A,y)=y\land y\in(S\setminus(A\setminus\inp(M,A)))^*}}\sum_{z\in(\red(A)^{-1}(y))_\perp}\P(z)R(z)\label{eq:monab-3}
        \\&=\sum_{\substack{y\in(\red(B)^{-1}(x))_\perp\colon\\\red(A,y)=y}}\sum_{z\in(\red(A)^{-1}(y))_\perp}\P(z)R(z)\label{eq:monab-4}
        \\&=\sum_{z\in(\red(B)^{-1}(x))_\perp}\P(z)R(z)\label{eq:monab-5}
        \\&=\P_B(x)R_B(x).\label{eq:monab-6}
    \end{align}
    \end{subequations}
    Here,~\eqref{eq:monab-1} follows from~\Cref{def:MRM-path-abstr} for the abstraction of $\abstr(A,\M)$ with respect to $B$.
    The fact that~\eqref{eq:monab-2} holds is seen as follows.
    Given $y\in(\red(B)^{-1}(x))_\perp$ we have $\red(B,y)=x$ so $|y|>1$ as $|x|>1$ (as $\red$ always outputs a word of length at most that of its second argument),
    meaning that $\P_A(y)=0$ if $\red(A,y)\neq y$ (in which case $\red(A)^{-1}(y)=\varnothing$ as $\red(A)$ is idempotent) but also $\P_A(y)=0$ if $y\notin(S\setminus(A\setminus\inp(M,A)))$ which follows directly from~\Cref{def:Markov-chain-path-abstr}.
    Then,~\eqref{eq:monab-3} holds because if $y\in(\red(B)^{-1}(x))_\perp$ is such that $\red(A,y)=y\land y\in(S\setminus(A\setminus\inp(M,A)))^*$
    we have $\P_A(y)R_A(y)=\sum_{z\in(\red(A)^{-1}(y))_\perp}\P(z)R(z)$.
    Indeed: if $\P_A(y)>0$ then this is immediate from~\Cref{def:MRM-path-abstr}, but if $\P_A(y)=0$ then $\P_A(y)R_A(y)=0$ but we also have $|y|>1$ and $y\in(S\setminus(A\setminus\inp(M,A)))^*$, so $\P_A(y)=0$ could only have been the result from $\P(\red(A)^{-1}(y))=0$,
    which means that $\P(z)=0$ for all $z\in(\red(A)^{-1}(y))_\perp$ due to the definition of probability mass of languages as given in~\eqref{eq:prob-mass}
    and hence $\sum_{z\in(\red(A)^{-1}(y))_\perp}\P(z)R(z)=0$ as well.
    Then,~\eqref{eq:monab-4} holds because of~\Cref{lem:spades}: if $y\in(\red(B)^{-1}(x))_\perp$ satisfies $\red(A,y)=y$ but $y\notin(S\setminus(A\setminus\inp(M,A)))^*$, then $\P(\red(A)^{-1}(y))=0$ by~\Cref{lem:spades} which is applicable because $x\in(S\setminus(B\setminus\inp(M,B)))^*$ as $\P_B(x)>0$ and $|x|>1$---cf.~\Cref{def:Markov-chain-path-abstr}---meaning that $\P(z)=0$ for all $z\in(\red(A)^{-1}(y))_\perp$ so $\sum_{z\in(\red(A)^{-1}(y))_\perp}\P(z)R(z)=0$.
    Finally,~\eqref{eq:monab-5} follows directly from~\Cref{thm:inverse-images-red-union},
    and~\eqref{eq:monab-6} is by definition of $R_B$ (\Cref{def:MRM-path-abstr}).
\qed
\end{proof}

\begin{remark}
    Adding or removing \emph{isolated}\footnote{A state $s\in S^*$ is called \emph{isolated} if $rs$ and $st$ have probability zero for all $r,t\in S$.} states to or from an abstraction set has no effect on the abstraction.
    Since we do not draw isolated states, \Cref{thm:MRM-monab} in practice implies that, whenever we abstract over a superset of states that have ever been coloured gray \emph{up to now}, then we obtain the same MRM as if we immediately abstracted over the set of all states that will in the end have been coloured.
    In ~\Cref{subfig:M,subfig:M-1,subfig:M-12,subfig:M-123} this means that
    $\M_{1,2,3}=\abstr(\set{s_1,\ldots,s_6},\M)$, thus that $\M_{1,2,3}$ is the result of fully model checking $\M$, even though we never abstracted fully over $\set{s_1,\ldots,s_6}$ to compute $\M_{1,2,3}$.
\end{remark}

\section{Numerical Recipe}
\label{sec:num-recipe}

We now provide a numerical recipe for computing path abstractions of MRMs in terms of solving sets of linear equations, which we turn into a high-level reference implementation in the computer algebra system PARI/GP~\cite{PARI2} in \Cref{sec:ref_impl}.

Throughout, we fix an MRM $\M=(S,s_0,\P,R)$ and write $M=(S,s_0,\P)$ for the underlying Markov chain of $\M$.
Furthermore, we fix a set $A\subseteq S$ and consider the abstracted MRM $\M'=(S,s_0,\P',R')=\abstr(A,\M)$.
We determine $R'(st)$ for all $s,t\in S$, which gives the full function $R'\colon S^* \to \RR$:
if a reward function is known on all state sequences of length $2$, then it is known on all (non-empty) state sequences---conforming to the second axiom in~\Cref{def:MRM-path-abstr}.

First, in~\Cref{subsec:num-recipe-results}, we state the main results for the correctness of the reference implementation in~\Cref{listing:ref-impl}; then, we provide the proofs in~\Cref{subsec:num-recipe-proofs}.

\subsection{Results}
\label{subsec:num-recipe-results}

Throughout, we write $\I\defeq\inp(M,A)=\set{s\in A\colon s=s_0\lor\P((S\setminus A)s)>0}$ for the set of input states of $A$ as in~\Cref{def:Markov-chain-path-abstr}, and $\O\defeq\set{s\in S\setminus A\colon\P(As)>0}$ for the set of states outside of $A$ that are reachable from $A$ in one step.
Let $\U\defeq\set{r\in A\colon\P(rA^*\O)>0}$
be the set of all states in $A$ from which we can reach states outside of $A$,
and $\U_1\defeq\set{r\in A\colon \P(r\O)>0}$ the set of all states in $A$ from which we can reach the outside of $A$ in one step.
Note that $\U_1\subseteq\U$.
Furthermore, $\U$ can be computed efficiently by means of a standard reachability algorithm, as shown in~\cite[Lemma 3]{HM25}.
We now express $R'$ fully in terms of the aforementioned state sets,
on state sequences of length $2$.

\begin{lemma}
\label{lem:R'-2state}
    Let $s,t\in S$. Then, we have
    \begin{equation}
    \label{eq:R'-2state}
        R'(st)=\begin{cases}
            R(st),&\text{if }s\in S\setminus A\land t\in S\setminus(A\setminus\I),\\
            \frac{1}{\P'(st)}\sum_{y\in sA^*t}\P(y)R(y),&\text{if }s\in\I\cap\U\land t\in\O\land\P'(st)>0,\\
            0,&\text{otherwise}.
        \end{cases}
    \end{equation}
\end{lemma}
We are thus interested in the second case of~\eqref{eq:R'-2state}.
We provide a numerical recipe in terms of solving a system of linear equations to calculate $R'(st)$ in that case; the set of this system's unknowns is exactly $\U$ and, upon solved, can be used for all choices of $s\in\I\cap\U$ and $t\in\O$ simultaneously.
We assume that $\O$ is non-empty---otherwise, the computation of $R'$ on $2$-state sequences is trivial.

To this end, we let $T\in\RR^{\U\times\U}$ denote the matrix carrying the probabilities of the Markov chain $M=(S,s_0,\P)$ between states in $\U$; thus, $T(u,v)=\P(uv)$.
For $i,j\in\U$, we let $T_{ij}\in\RR^{\U\times\U}$ denote the matrix obtained from $T$ by $T_{ij}(u,v)=T(u,v)$ if $ij=uv$ and $T_{ij}(u,v)=0$ otherwise---so, $T=\sum_{i,j\in\U}T_{ij}$.
Finally, since $T$ has spectral radius $<1$, as argued in (the proof of)~\cite[Thm. 5]{HM25arxivversion},
we can sum all its powers $Q\defeq\1+T+T^2+\ldots$ and obtain the matrix $Q\in\RR^{\U\times\U}$ by inverting the matrix $\1-T$. In other words, $Q=(\1-T)^{-1}$.

We now state the main result relating path abstraction to pure linear algebra.

\begin{theorem}
\label{thm:main-num-result}
    If $s\in\I\cap\U$ and $t\in\O$ and $\P'(st)>0$, then we have
    \begin{subequations}\label{eq:main-num-result}
    \begin{align}
    \P'(st)R'(st)&=\sum_{r\in\U_1}\sum_{i,j\in\U\colon R(ij)\neq0}(QT_{ij}Q)(s,r)\P(rt)R(ij)\label{eq:main-num-result-1}
        \\&+Q(s,s)\P(st)R(st)\label{eq:main-num-result-2}
        \\&+\sum_{i\in\U\setminus\set{s}\colon R(it)\neq0}(Q-\1)(s,i)\P(it)R(it).\label{eq:main-num-result-3}
    \end{align}
    \end{subequations}
\end{theorem}
Before we prove~\Cref{thm:main-num-result}, we establish two intermediate results.
The first enables us to rewrite the sum in the second case of~\eqref{eq:R'-2state} as a linear combination of rewards of just \emph{transitions}.
Thus, fix $s\in\I\cap\U$ and $t\in\O$ and suppose that $\P'(st)>0$.
We set, for all $i\in\U$ and $j\in\Ut$,\!\footnote{Here, $\evt_{ij}$ is the EVT, the \emph{expected number of visits}, of $ij$ in the sub-MRM.}
\begin{equation}
\label{eq:evt-def}
\evt_{ij}\defeq\sum_{y\in s\U^*t}\P(y)\cdot|y|_{ij},
\end{equation}
where $|y|_{ij}$ is the number of times $ij$ appears as a factor in $y$, cf.~\Cref{def:factor-occurrences}.

\begin{lemma}
\label{lem:rew-ito-evt}
    We have
    \begin{equation}
    \label{eq:rew-ito-evt}
        \P'(st)R'(st)=\sum_{i\in\U\land j\in\Ut}\evt_{ij}\cdot R(ij).
    \end{equation}
\end{lemma}
Now, for all $i\in\U$ and $j\in\Ut$ and $l\in\NN^+$, we set
\begin{equation}
    \label{eq:evt-l-def}
    \evt_{ij}^{(l)}\defeq\sum_{y\in s\U^{l-1}t}\P(y)\cdot|y|_{ij},
\end{equation}
and note that
\begin{equation}
\label{eq:evt-ito-evt-l}
\evt_{ij}=\sum_{l=1}^\infty\evt_{ij}^{(l)}.
\end{equation}
To eliminate the factors $|y|_{ij}$ in the summation in~\eqref{eq:evt-l-def}, we express each $\evt_{ij}^{(l)}$ in terms of probability masses of (potentially) overlapping languages as follows.

\begin{lemma}
\label{lem:evts-ito-prob-masses}
    Let $i\in\U$ and $j\in\Ut$. Then we have
    \begin{subequations}\label{eq:evts-ito-prob-masses}
    \begin{align}
        \evt_{ij}^{(l)}&=\sum_{a=0}^{l-1}\P(s\U^{l-1}t\cap S^aijS^{l-a-1})\quad\text{for all }l\in\NN^+,\label{eq:evts-ito-prob-masses-1}
        \\\evt_{ij}&=\begin{cases}
            \sum_{r\in\U_1}\sum_{a=0}^\infty\sum_{b=0}^\infty
            \P(s\U^{a+b}r\cap S^aijS^b)\P(rt),&i,j\in\U,\\
            \P(st)+\P(s\U^*st),&ij=st,\\
            \sum_{y\in s\U^*i}\P(y)\P(it),&i\neq s,j=t.
        \end{cases}\label{eq:evts-ito-prob-masses-2}
    \end{align}
    \end{subequations}
\end{lemma}

\subsection{Proofs}
\label{subsec:num-recipe-proofs}

We need to prove~\Cref{lem:R'-2state,thm:main-num-result,lem:rew-ito-evt,lem:evts-ito-prob-masses}.
First, \Cref{lem:R'-2state} is easy to prove following its case distinction---see \ifarxivelse{\Cref{sec:AppProofs}}{App.\ A}.
We then use it to prove~\Cref{lem:rew-ito-evt,lem:evts-ito-prob-masses}; ultimately, we use~\Cref{lem:rew-ito-evt,lem:evts-ito-prob-masses} to prove~\Cref{thm:main-num-result}.

\begin{proof}[of~\Cref{lem:rew-ito-evt}]
    In the context of~\Cref{lem:rew-ito-evt}, we had already fixed $s\in\I\cap\U$ and $t\in\O$ and assumed $\P'(st)>0$.
    We are thus in the second case of~\eqref{eq:R'-2state} which we know holds because we just proved~\Cref{lem:R'-2state}.
    We thus have
    \begin{subequations}\label{eq:rew-ito-evt-proof}
    \begin{align}
        \P'(st)R'(st)
        &=\sum_{y\in sA^*t}\P(y)R(y)\label{eq:rew-ito-evt-proof-1}
        \\&=\sum_{y\in s\U^*t}\P(y)R(y)\label{eq:rew-ito-evt-proof-2}
        \\&=\sum_{y\in s\U^*t}\P(y)\sum_{i\in\U\land j\in\Ut}|y|_{ij}\cdot R(ij)\label{eq:rew-ito-evt-proof-3}
        \\&=\sum_{i\in\U\land j\in\Ut}\left(\sum_{y\in s\U^*t}\P(y)\cdot|y|_{ij}\right)R(ij)\label{eq:rew-ito-evt-proof-4}
        \\&=\sum_{i\in\U\land j\in\Ut}\evt_{ij}\cdot R(ij),\label{eq:rew-ito-evt-proof-5}
    \end{align}
    where~\eqref{eq:rew-ito-evt-proof-1} follows from the second case of~\eqref{eq:R'-2state};~\eqref{eq:rew-ito-evt-proof-2} follows from the definition of $\U$ in~\Cref{subsec:num-recipe-results};~\eqref{eq:rew-ito-evt-proof-3} follows from the fact that, given $y\in s\U^*t$, $R(y)$ is the sum of the rewards on all $2$-state factors contained in $y$ which can be calculated as the sum over all $2$-state sequences $ij\in\U(\Ut)$ of the reward $R(ij)$ times the number $|y|_{ij}$ of times $ij$ occurs as a factor in $y$---cf.~\Cref{def:factor-occurrences}---as all possible $2$-state factors of $y\in s\U^*t$ are in $\U(\Ut)$ because $s\in\U\not\ni t$;~\eqref{eq:rew-ito-evt-proof-4} follows from changing the summation order; and~\eqref{eq:rew-ito-evt-proof-5} is by definition of $\evt_{ij}$ in~\eqref{eq:evt-def}.
    \end{subequations}
    \qed
\end{proof}

\vspace{-16pt}
\begin{proof}[of~\Cref{lem:evts-ito-prob-masses}]
    Let $i\in\U$ and $j\in\Ut$.

    We first prove~\eqref{eq:evts-ito-prob-masses-1}.
    Thus, let $l\in\NN^+$.
    Starting from~\eqref{eq:evt-l-def}, we obtain
    \begin{equation}
    \label{eq:evt-l-deriv}
        \evt_{ij}^{(l)}=\sum_{k=1}^lk\cdot\P(\set{y\in s\U^{l-1}t\colon|y|_{ij}=k})
    \end{equation}
    because $|y|=l+1$ for all $y\in s\U^{l-1}t$ so no $y\in s\U^{l-1}t$ can contain $ij$ as a factor more than $l$ times.
    Now, to the sum of~\eqref{eq:evt-l-deriv}, any $y\in s\U^{l-1}t$ contributes its probability mass exactly $|y|_{ij}$ times and hence is in exactly $|y|_{ij}$ of the sets $s\U^{l-1}t\cap S^aijS^{l-a-1}$ where $a=0,\ldots,l-1$.
    Conclude that the right-hand side of~\eqref{eq:evt-l-deriv}
    is equal to the right-hand side of~\eqref{eq:evts-ito-prob-masses-1}, and thus that~\eqref{eq:evts-ito-prob-masses-1} holds.

    Next, we prove~\eqref{eq:evts-ito-prob-masses-2}, of which we first handle the ``hard'' case $i,j\in\U$.
    Since $j\in\U$ but $t\notin\U$, it follows that $\evt_{ij}^{(1)}=0$ and $(\clubsuit)$ for all integers $l\geqslant2$ we have $\evt_{ij}^{(l)}=\sum_{a=0}^{l-2}\P(s\U^{l-1}t\cap S^aijS^{l-a-1})$ using~\eqref{eq:evts-ito-prob-masses-1}. It follows that
    \begin{subequations}\label{eq:evt-deriv-hard-case}
    \begin{align}
        \evt_{ij}
        &=\sum_{l=1}^\infty\evt_{ij}^{(l+1)}\label{eq:evt-deriv-hard-case-1}
        \\&=\sum_{l=1}^\infty\sum_{a=0}^{l-1}\P(s\U^lt\cap S^aijS^{l-a})\label{eq:evt-deriv-hard-case-2}
        \\&=\sum_{l=1}^\infty\sum_{a=0}^{l-1}\sum_{r\in\U_1}\P(s\U^{l-1}r\cap S^aijS^{l-a-1})\P(rt)\label{eq:evt-deriv-hard-case-3}
        \\&=\sum_{r\in\U_1}\sum_{a=0}^\infty\sum_{b=0}^\infty\P(s\U^{a+b}r\cap S^aijS^b)\P(rt),\label{eq:evt-deriv-hard-case-4}
    \end{align}
    \end{subequations}
    where~\eqref{eq:evt-deriv-hard-case-1} follows from~\eqref{eq:evt-ito-evt-l} and the fact that $\evt_{ij}^{(1)}=0$;~\eqref{eq:evt-deriv-hard-case-2} follows from $(\clubsuit)$;~\eqref{eq:evt-deriv-hard-case-3} follows by definition of $\U_1$ in~\Cref{subsec:num-recipe-results}; and~\eqref{eq:evt-deriv-hard-case-4} holds by the substitution $b=l-a-1$ and a change of the summation order.

    Finally, we verify the latter two cases of~\eqref{eq:evts-ito-prob-masses-2}. To this end, we notice that in both these cases we have $j=t$ so that each $y\in s\U^*t$ satisfies $|y|_{ij}\leqslant1$ as $s\in\U\not\ni t$,
    and thus that $\evt_{ij}=\P(\set{y\in s\U^*t\cap S^*it})$---cf.~\eqref{eq:evt-def}.
    In case $i=s$, which is the second case of~\eqref{eq:evts-ito-prob-masses-2}, we have $\P(\set{y\in s\U^*t\cap S^*it})=\P(\set{y\in s\U^*t\cap S^*st})=\P(st)+\P(s\U^*st)$, proving that case;
    in case $i\neq s$, which is the third case of~\eqref{eq:evts-ito-prob-masses-2}, we have $\P(\set{y\in s\U^*t\cap S^*it})=\P(s\U^*it)=\sum_{y\in s\U^*t}\P(y)\P(it)$, proving that case as well.
    \qed
\end{proof}

\begin{proof}[of~\Cref{thm:main-num-result}]
    Note that to prove~\Cref{thm:main-num-result} it is, by~\eqref{eq:rew-ito-evt} and~\eqref{eq:evts-ito-prob-masses-2},
    sufficient to prove that the three contributions in~\eqref{eq:main-num-result-1},~\eqref{eq:main-num-result-2}, and~\eqref{eq:main-num-result-3} correspond to summing the $\evt_{ij}$'s in the three cases of~\eqref{eq:evts-ito-prob-masses-2}, respectively.

    Thus, first we prove that $(QT_{ij}Q)(s,r)=\sum_{a=0}^\infty\sum_{b=0}^\infty\P(s\U^{a+b}r\cap S^aijS^b)$
    for all $i,j\in\U$ and $r\in\U_1$. This is seen as $Q=\1+T+T^2+\ldots$, so given $a,b\in\NN$ it suffices to show that $(T^aT_{ij}T^b)(s,r)=\P(s\U^{a+b}r\cap S^aijS^b)$.
    Clearly, as $s,r\in\U$, $\P(s\U^{a+b}r)$ is the probability mass of all sequences of length $a+b+2$ starting in $s$, moving through $\U$, and ending in $r$, which is classically given by the $(s,r)$-th entry of $T^{a+b+1}$, as $T\in\RR^{\U\times\U}$ was defined to encode the probability masses of transitions between states in $\U$ via $T(u,v)\defeq\P(uv)$; but now, we have to replace the $(a+1)$-th matrix in the $(a+b+1)$-fold product of $T$ with itself by $T_{ij}$---where we recall that $T_{ij}(u,v)=\P(uv)$ if $uv=ij$ and zero otherwise---to calculate $\P(s\U^{a+b}r\cap S^aijS^b)$, as this is the probability mass of all sequences in $s\U^{a+b}r$ whose $(a+1)$-th $2$-state factor is forced to be $ij$.

    Then, we prove that $Q(s,s)=1+\sum_{y\in s\U^*s}\P(y)$.
    Since $Q=\1+T+T^2+\ldots$ we have $Q(s,s)=1+T(s,s)+T^2(s,s)+\ldots=1+\sum_{l=1}^\infty\sum_{y\in s\U^{l-1}s}\P(y)=1+\sum_{y\in s\U^*s}\P(y)$, which proves the claim.

    Finally, we show that for $i\in\U\setminus\set{s}$,
    we have $(Q-\1)(s,i)=\sum_{y\in s\U^*i}\P(y)$.
    We have $(Q-\1)(s,i)=T(s,i)+T^2(s,i)+\ldots=\sum_{l=1}^\infty\sum_{y\in s\U^{l-1}i}\P(y)=\sum_{y\in s\U^*i}\P(y)$,
    and we are done.
    \qed
\end{proof}

\section{Reference Implementation}
\label{sec:ref_impl}

\lstdefinelanguage{gp}
{
    morecomment=[l]{\\},
    morecomment=[s]{/*}{*/},
    morekeywords={if,while,prod,sum, concat,for, vecmin,factorint, matsize,gcd, vector, break, matrix, matid, matsolve, vecextract, vecsum, pathAbstr,pathAbstrMRM}
}
\definecolor{lightgray}{rgb}{0.95,0.95,0.95}
\definecolor{darkgray}{rgb}{0.3,0.3,0.3}
\lstset{
    classoffset=0,
    language={gp},
    breaklines=true,
    basicstyle=\fontsize{7.75pt}{9.29pt}\ttfamily,
    linewidth=\textwidth,
    backgroundcolor=\color{white},
    captionpos=b, %
    extendedchars=true, %
    tabsize=2, %
    columns=fixed, %
    xleftmargin=0.0\textwidth,
    xrightmargin=0.0\textwidth,
    keepspaces=true, %
    showstringspaces=false, %
    breaklines=true, %
    frame=rl  bt, %
    keywordstyle=\color{black}\bfseries,%
    framesep=4pt, %
    numbers=left, %
    numberstyle=\tiny\ttfamily, %
    commentstyle=\color{gray}, %
    stringstyle=\color{red}, %
    identifierstyle=\color{black},
}
\counterwithout{lstlisting}{section}
\begin{lstlisting}[float=!tb, abovecaptionskip=8pt, label={listing:ref-impl}, caption={An implementation in PARI/GP~\cite{PARI2} of path abstraction for MRMs}]
pathAbstrMRM = A -> (MRM ->  \
	P = MRM[1]; R = MRM[2]; n = matsize(P)[1]; \
	II = vector(n, i, A[i] && (i == 1 || \
		vecsum(vector(n, j, !A[j] * P[j, i])) != 0)); \
	OO = vector(n, i, !A[i] && \
		vecsum(vector(n, j, A[j] * P[j, i])) != 0); \
	Uvec = [OO]; Vvec = [OO]; \
	i = 1; \
	while (Uvec[i] != vector(n), \
		Uvec = concat(Uvec, [vector(n, j, A[j] && !Vvec[i][j] && \
			vecsum(vector(n, k, Uvec[i][k] * P[j, k])) != 0)]); \
		Vvec = concat(Vvec, [vector(n, j, Vvec[i][j] || Uvec[i+1][j])]); \
		i++; \
	); \
	U1 = Uvec[2]; U = vector(n, j, Vvec[i][j] && !OO[j]); \
	PU = matrix(n, n, i, j, (U[i] && U[j]) * P[i, j]); \
	Q = matadjoint(matid(n) - PU) / matdet(matid(n) - PU); \
	P_A = matrix(n, n, i, j, if(!A[i] && (!A[j] || II[j]), \
		P[i, j], if(II[i] && U[i] && OO[j], Q[i,] * P[, j], 0))); \
	R_A = matrix(n, n, i, j, if(!A[i] && (!A[j] || II[j]), R[i, j], 0)); \
	for (h = 1, n, if (II[h] && U[h], \
		for (j = 1, n, if (OO[j] && P_A[h, j] != 0, \
			R_A[h, j] = (1 / P_A[h, j]) * (sum(i = 1, n, \
				if(U1[i], sum(a = 1, n, \
				if(U[a], sum(b = 1, n, \
				if(U[b] && R[a, b] != 0, \
				(Q*matrix(n, n, i0, j0, i0 == a && j0 == a) \
					* PU * matrix(n, n, i0, j0, i0==b && j0==b)*Q)[h, i] \
					* P[i, j]*R[a, b], 0)), 0)), 0)) \
				+ Q[h, h] * P[h, j] * R[h, j] \
				+ sum(a = 1, n, if(U[a] && a != h && R[a, j] != 0, \
				(Q-matid(n))[h, a]*P[a, j]*R[a, j], 0))); \
		)); \
	)); \
	[P_A, R_A]; \
);
\end{lstlisting}

We include a high-level reference implementation of path abstraction for MRMs in the computer algebra system PARI/GP~\cite{PARI2}, extending~\cite[Sect.\ 5: Listing~1]{HM25} from reachability probabilities for Markov chains to ERs for MRMs. %
MRMs are represented as a vector containing a probability matrix and a matrix carrying the transition rewards.
The initial state is $1$. Columns and rows corresponding to states that do not play a role (anymore) are fully zero.

The inner workings of~\Cref{listing:ref-impl} are as follows, following the notation in~\Cref{sec:num-recipe}.
The ``hard'' case for the computation of $R'(st)$ in~\eqref{eq:R'-2state} and expressed in terms of linear algebra in~\Cref{thm:main-num-result} is implemented in~\Cref{listing:ref-impl} in lines 22-33.
All that the rest of the reference implementation is handling is the computation of the sets $\I$, $\O$, $\U$, and $\U_1$ (\verb!II!, \verb!OO!, \verb!U!, and \verb!U1!) in lines 3-16, of the matrices $T$ and $Q$ (\verb!PU! and \verb!Q!) in lines 17-18, and the handling of the easy cases for $R'(st)$ (cf. cases 1 and 3 in~\eqref{eq:R'-2state}) in lines 19-21.
The remaining lines are 34-37 (ends of loops and output of the routine), and lines 1-2 (header of routine, plus initializing the probability- and reward matrices \verb!P! and \verb!R!).
Note that the hardest computational step, that of inverting the matrix $\1-T$ on line 18, only has to be done once; then, we can simply iterate over all combinations of $s\in\I\cap\U$ and $t\in\O$ with $\P'(st)>0$ (lines 22-23) to compute $R'$ while (re-)using the same $Q$.

\section{Conclusion and Future Work}
\label{sec:Conclusion}

We extended path abstraction from reachability probabilities on DTMCs to expected rewards on MRMs.
Working throughout in the free monoid view, we proved that path abstracting an MRM over an arbitrary set of states again yields an MRM and that it remains monotonically absorbing.
We gave a numerical recipe that reduces one abstraction to solving one linear equation system, provided a reference implementation in PARI/GP, and proved it correct by way of the connection between expected rewards and visiting times.

Several directions remain for future work,
the most immediate being an efficient implementation.
This could be a standalone tool that reads a Markov chain in the UMB format~\cite{AHJ+26} plus an abstraction set and writes the result back to a file, or a model checking engine inside an established tool such as \tool{mcsta} of the \toolset~\cite{HH14} that solves the arising linear equation systems with (possibly exact) methods.
For the latter to be effective, we must find good heuristics for choosing abstraction sets, ideally improving on the SCC-based decomposition of \'{A}brah\'{a}m \etal~\cite{AJWKB10}.
A second direction is to extend our work towards continuous-time Markov chains with time-bounded properties.
Finally, we would like to formalise the free monoid view of Markov chains, path abstraction, and the theorems of this paper and~\cite{HM25} in an interactive theorem prover such as Isabelle/HOL, to machine-check our proofs and ultimately derive a correct-by-construction implementation, as recently done for interval iteration~\cite{KSAHL25}.

\bibliographystyle{splncs04}
\bibliography{paper}

\iftoggle{arxiv}{%
\clearpage
\appendix
\crefalias{section}{appendix} %
\crefalias{subsection}{appendix} %

\section{Additional Proofs}
\label{sec:AppProofs}

\begin{proof}[of~\Cref{lem:R'-2state}]
    If $s\in S\setminus A$ and $t\in S\setminus(A\setminus\I)$,
    then we have $\red^{-1}(A)(st)=stA^*$ if $t\in \I$ and $\red^{-1}(A)(st)=st$ otherwise---cf.~\Cref{ex:red-2state-char}.
    In both cases, $(\red(A)^{-1}(st))_\perp=st$.
    Since $s,t\notin A\setminus\I$, we are in the second case of~\eqref{eq:Markov-chain-path-abstr} so $\P'(st)=\P(st)$.
    If $\P'(st)=\P(st)=0$ then $R'(st)=0$ and $R(st)$ because rewards are zero on state sequences with probability zero---cf.~\Cref{def:MRM}---and if $\P'(st)=\P(st)>0$ then by~\eqref{eq:reward-function} we obtain $R'(st)=R(st)$ directly as $|st|>1$ and $(\red^{-1}(A)(st))_\perp=\set{st}$.

    Now, if $s\in\I\cap\U$ and $t\in\O$ satisfy $\P'(st)>0$,
    then we are in the first case of~\eqref{eq:reward-function}.
    Note that here $\red(A)^{-1}(st)=sA^*t$ and $(sA^*t)_\perp=sA^*t$ because $s\in A\not\ni t$.
    It follows that $R'(st)=\frac{1}{\P'(st)}\sum_{y\in sA^*t}\P(y)R(y)$.

    In all other cases, it is straightforward to verify that $R'(st)=0$.
    \qed
\end{proof}

}{}

\end{document}